\documentclass[12pt]{amsart}
\usepackage{graphicx}
\usepackage[headings]{fullpage}
\usepackage{amssymb,epic,eepic,epsfig,amsbsy,amsmath,amscd,color,bm}
\numberwithin{equation}{section}
                        \theoremstyle{plain}
\usepackage{mathrsfs}

\newcommand\no[1]{}

\newtheorem{theorem}{Theorem}[section]
\newtheorem{thm}{Theorem}
\newtheorem{lemma}[theorem]{Lemma}

\theoremstyle{definition}
\newtheorem{remark}[theorem]{Remark}

\def\CH{\mathcal H}

\def\CS{\mathcal S}
\def\CT{\mathcal T}

\def\be { \begin{equation} }
\def\ee { \end{equation} }

\begin{document}
\allowdisplaybreaks
\baselineskip16pt
\title[An upper bound for the purity of APPT states]{An upper bound for the purity of absolutely positive partial transpose  states}

\author{Anh T. Tran}
\address{Department of Mathematical Sciences, The University of Texas at Dallas, Richardson, TX 75080, USA}
\email{att140830@utdallas.edu}


\begin{abstract}
A quantum state is called absolutely separable (resp. absolutely positive partial transpose (APPT)) if it remains separable (resp. positive partial transpose (PPT)) under any global unitary transformation. It is known that the set of all absolutely separable states is a subset of the set of all APPT states. Moreover, these two sets are identical for qubit-qudit systems. 
In this note, we give an upper bound for the purity of APPT bipartite states (and therefore for the purity of absolutely separable bipartite states). For qubit-qudit systems, this upper bound becomes the maximum purity of APPT (and absolutely separable) states. 
\end{abstract}
\maketitle

\section{Main result}

Quantum entanglement is a fundamental phenomenon in quantum mechanics where two or more particles become interconnected such that the quantum state of each particle cannot be described independently of the state of the others, even when the particles are separated by large distances. Entanglement is central to understanding the distinction between classical and quantum theories. Moreover, it serves as a core resource for parallel quantum processing and quantum algorithms. 

A bipartite quantum state $\rho$ on a Hilbert space $\CH_A \otimes \CH_B$ is called separable if  it can be written as a convex combination  of product states, namely, $\rho =\sum_i p_i \, \rho_i^A \otimes \rho_i^B$ where $p_i \ge 0$, $\sum_i p_i=1$, and $\rho_i^A$ and $\rho_i^B$ are density operators on subsystems $A$ and $B$ respectively. An entangled state is a quantum state that is not separable. Peres \cite{Pe} showed that any separable state $\rho$ has a positive partial transpose (PPT), meaning that the density matrix of $\rho$ remains positive semi-definite after taking the partial transpose with respect to one of its subsystems. 

Following \cite{KZ} and \cite{Hi}, we say that a state $\rho$ on $\CH_A \otimes \CH_B$ is absolutely separable (resp. APPT) if it remains separable (resp. PPT) under any global unitary transformation, namely, $U \rho U^\dagger$ is separable (resp. PPT) for any unitary operator $U: \CH_A \otimes \CH_B \to \CH_A \otimes \CH_B$. Then it follows from \cite{Pe} that the set of all absolutely separable states is a subset of the set of all APPT states. Johnston \cite{Jo} proved that these two sets are identical for qubit-qudit systems (the case when $\CH_A$ has dimension $2$). However, in general, it is still unknown whether they are identical or not. 

Hildebrand \cite{Hi} provided  a necessary and sufficient condition for a bipartite state $\rho$ to have the APPT property. This condition can be expressed as a set of linear matrix inequalities on the eigenvalues of the density matrix of $\rho$. In particular, the inequality $\lambda_1 \le \lambda_{mn-1} + 2\sqrt{\lambda_{mn-2} \lambda_{mn}}$ always holds true for any APPT state $\rho$ on $\CH_A \otimes \CH_B$, where
 $\lambda_1 \ge \lambda_2 \ge \cdots \ge \lambda_{mn}$ are the eigenvalues of the density matrix of $\rho$ in the decreasing order, and 
$m$ and $n$ are respectively the dimension of $\CH_A$ and $\CH_B$. In general, the set of linear matrix inequalities in \cite{Hi} can be complicated, but it becomes particularly simple in the case of qubit-qudit systems (i.e. $m=2$). More precisely, for qubit-qudit systems the inequality $\lambda_1 \le \lambda_{2n-1} + 2\sqrt{\lambda_{2n-2} \lambda_{2n}}$ is exactly  the necessary and sufficient condition for a bipartite state $\rho$ to have the APPT property. Regarding absolute separation, Gurvits and Barnum \cite{GB} proved that a state $\rho$ on $\CH_A \otimes \CH_B$  is absolutely separable if its purity satisfies $\mathrm{Tr} (\rho^2) \le \frac{1}{mn-1}$. However, finding a necessary and sufficient condition for $\rho$ to be absolutely separable is still a challenging problem. 

We now discuss the purity $\mathrm{Tr} (\rho^2)$ of absolutely separable and APPT states $\rho$ on $\CH_A \otimes \CH_B$. Kondra, Hita, Neumann, Kampermann and Bruß \cite[Lemma 5]{K+} showed that the purity of any absolutely separable state $\rho \in \CH_A \otimes \CH_B$, where $mn \ge 6$, is bounded above by $\frac{2}{mn}$. For qubit-qudit systems (i.e. $m=2$), Song and Chen \cite[Corollary 11]{SC} proved that the maximum purity of APPT (and absolutely separable) states is equal to $\frac{2}{3n}$ if $n$ is even and $n \ge 4$, and equal to $\frac{6n+4}{(3n+1)^2}$ if $n$ is odd. They also proved that in the case of two-qubit systems (i.e. $m=n=2$), 
the maximum purity of APPT (and absolutely separable) states is equal to $\frac{3}{8}$. See also \cite[Theorem 3.1]{DK}. 

Our main result of this note is the following upper bound for the purity of APPT states (and therefore for the purity of absolutely separable bipartite states). For qubit-qudit systems (i.e. $m=2$), this upper bound is exactly the maximum purity of APPT (and absolutely separable) states determined in \cite{SC}.

\begin{thm}
The purity of any absolutely positive partial transpose state $\rho$ on $\CH_A \otimes \CH_B$, where $mn \ge 8,$ is bounded above by 
$$
\begin{cases}
\frac{4}{3mn} & \text{if $mn \equiv 0 \pmod{4}$}, \medskip \\

	                 \frac{4(3mn-2)}{(3mn-1)^2} & \text{if $mn \equiv 1 \pmod{4}$}, \medskip \\
	                 
	                 \frac{4(3mn + 4)}{(3mn+2)^2}  & \text{if $mn \equiv 2 \pmod{4}$}, \medskip\\
	                 
	               \frac{4(3mn + 2)}{(3mn+1)^2} & \text{if $mn \equiv 3 \pmod{4}$}. 
\end{cases}
$$
\end{thm}

\begin{remark}
For qutrit-qudit systems (i.e. $m=3$), D\~ung and Khoi \cite[Conjecture 3.2]{DK} predicted the upper bound in Theorem $1$. However, it should be noted that this upper bound is not equal to the maximum purity of APPT states. It is an open problem to determine exactly the maximum purity of absolutely separable and APPT bipartite states, especially in the case of qutrit-qudit systems. 
\end{remark}

The rest of this note is devoted to proving Theorem $1$.

\section{Proof of Theorem $1$}

Suppose that $\rho$ is an APPT state on $\CH_A \otimes \CH_B$, with $m = \dim \CH_A$ and $n = \dim \CH_B$. Let $\lambda_1 \ge \lambda_2 \ge \cdots \ge \lambda_{mn}$ be the eigenvalues of the density matrix of $\rho$ in the decreasing order.
Then the purity of $\rho$ is given by $\mathrm{Tr}(\rho^2) = \sum_{i=1}^{mn} \lambda_i^2$.  

Note that $\mathrm{Tr}(\rho) = \sum_{i=1}^{mn} \lambda_i=1$. Since $\rho$ has the APPT property, by \cite{Hi} we have  $\lambda_1 \le \lambda_{mn-1} + 2\sqrt{\lambda_{mn-2} \lambda_{mn}}$. In particular, $\lambda_1 \le \lambda_{mn-2} + \lambda_{mn-1} +\lambda_{mn}$.

Let $\CS_{k}$ be the set of real sequences $\pmb{\lambda} = (\lambda_1, \lambda_2, \cdots, \lambda_{k})$ such that $\lambda_1 \ge \lambda_2 \ge \cdots \ge \lambda_{k} \ge 0$, $\sum_{i=1}^{k} \lambda_i =1$ and $\lambda_1 \le \lambda_{k-2} + \lambda_{k-1} +\lambda_{k}$. Note that $\CS_{k}$ is a convex polytope. 

Define a function $F$ on $\CS_k$ by $F(\pmb{\lambda}) = \sum_{i=1}^{k} \lambda_i^2$. We want to maximize $F(\pmb{\lambda})$. Since $F$ is a strictly convex function on the compact convex polytope $\CS_{k}$,  the maximum of $F$  on $\CS_{k}$ exists and this maximum is attained only at certain vertices of $\CS_{k}$. 

\begin{lemma}
Let $k \ge 4$. The vertices of $\CS_{k}$ are given as follows.
\begin{itemize}
\item Type $1$ vertices: $\Big( \underbrace{\frac{1}{i}, \cdots,  \frac{1}{i}}_{i}, \,  \underbrace{0, \cdots 0}_{k-i} \Big)$ where $k-2 \le i \le k$.
\item Type $2$ vertices:  $\Big( \underbrace{\frac{2}{k-1+i}, \cdots, \frac{2}{k-1+i}}_{i}, \underbrace{\frac{1}{k-1+i}, \cdots,  \frac{1}{k-1+i}}_{k-1-i},  0 \Big)$ where $1 \le i \le k-3$.
\item Type $3$ vertices:  $\Big( \underbrace{\frac{3}{k+2i}, \cdots, \frac{3}{k+2i}}_{i}, \underbrace{\frac{1}{k+2i}, \cdots,  \frac{1}{k+2i}}_{k-i} \Big)$ where $1 \le i \le k-3$.
\end{itemize} 
\end{lemma}

\begin{proof}
Let $\CT_{k}$ be  the set of real sequences $\pmb{\lambda} = (\lambda_1, \lambda_2, \cdots, \lambda_{k})$ such that $\lambda_1 \ge \lambda_2 \ge \cdots \ge \lambda_{k} \ge 0$ and $\sum_{i=1}^{k} \lambda_i =1$. Then $\CT_{k}$ is a convex polytope with vertices $\pmb{\nu}_i = \Big( \underbrace{\frac{1}{i}, \cdots,  \frac{1}{i}}_{i}, \,  \underbrace{0, \cdots 0}_{k-i} \Big)$, where $1 \le i \le k$. Moreover, any straight segment $\overline{\pmb{\nu}_i\pmb{\nu}_j}$, $i \not= j$, is an edge of $\CT_{k}$. 

Let $G(\pmb{\lambda}) = \lambda_{k-2} + \lambda_{k-1} +\lambda_{k} - \lambda_1$. Then $\CS_{k}=\{ \pmb{\lambda} \in \CT_{k} \mid G(\pmb{\lambda}) \ge 0\}$. 
Note that 
$$
G(\pmb{\nu}_i) = \begin{cases}
	  - \frac{1}{i} & \text{$1 \le i \le k-3$}, \smallskip \\
	  
            0 & \text{if $i=k-2$},  \smallskip \\
            
            \frac{1}{k-1}  & \text{if $i = k-1$}, \smallskip \\
            
            \frac{2}{k}  & \text{if $i = k$}.
\end{cases}
$$
Hence $\pmb{\nu}_i \in \CS_{k}$ if and only if $k-2 \le i \le k$. 

We now determine vertices of $\CS_k$ that do not come from vertices of $\CT_k$. According to the theory of polyhedra, such a vertex must be a point $\pmb{\lambda}$ on an edge $\overline{\pmb{\nu}_i\pmb{\nu}_j}$ satisfying $G(\pmb{\nu}_i) > 0 > G(\pmb{\nu}_j)$ and $G(\pmb{\lambda}) = 0$. Write $\pmb{\lambda} = t \pmb{\nu}_i + (1-t) \pmb{\nu}_j$, $0 < t < 1$, where $1 \le i \le k-3$ and $k-1 \le j \le k$. Then 
\begin{eqnarray*}
\pmb{\lambda} &=& t \Big( \underbrace{\frac{1}{i}, \cdots,  \frac{1}{i}}_{i}, \,  \underbrace{0, \cdots 0}_{k-i} \Big) + (1-t) \Big( \underbrace{\frac{1}{j}, \cdots,  \frac{1}{j}}_{j}, \,  \underbrace{0, \cdots 0}_{k-j} \Big) \\
&=& \Big( \underbrace{\frac{t}{i} + \frac{1-t}{j}, \cdots, \frac{t}{i} + \frac{1-t}{j}}_{i}, \underbrace{\frac{1-t}{j}, \cdots,  \frac{1-t}{j}}_{j-i},  \underbrace{0, \cdots 0}_{k-j} \Big)
\end{eqnarray*}
and so
$$
G(\pmb{\lambda}) = \begin{cases}
            \frac{1-t}{k-1} - \frac{t}{i}  & \text{if $j = k-1$}, \smallskip \\
            
            \frac{2(1-t)}{k} - \frac{t}{i}   & \text{if $j = k$}.
\end{cases}
$$ 
Hence $G(\pmb{\lambda}) = 0$ if and only if $$
t = \begin{cases}
            \frac{i}{k-1+i}  & \text{if $j = k-1$}, \smallskip \\
            
            \frac{2i}{k+2i}    & \text{if $j = k$}.
\end{cases}
$$ 
This is equivalent to 
$$
\pmb{\lambda}  = \begin{cases}
             \Big( \underbrace{\frac{2}{k-1+i}, \cdots, \frac{2}{k-1+i}}_{i}, \underbrace{\frac{1}{k-1+i}, \cdots,  \frac{1}{k-1+i}}_{k-1-i},  0 \Big)  & \text{if $j = k-1$}, \smallskip \\
             \Big( \underbrace{\frac{3}{k+2i}, \cdots, \frac{3}{k+2i}}_{i}, \underbrace{\frac{1}{k+2i}, \cdots,  \frac{1}{k+2i}}_{k-i} \Big)   & \text{if $j = k$}.
\end{cases}
$$ 
These vertices, together with $\pmb{\nu}_{k-2}$, $\pmb{\nu}_{k-1}$, $\pmb{\nu}_{k}$, are all vertices of $\CS_k$.
\end{proof}

\begin{theorem} \label{linear}
If $k \ge 8$ then the maximum of  $F(\pmb{\lambda}) $ on $\CS_{k}$ is 
$$
\begin{cases}
\frac{4}{3k} & \text{if $k \equiv 0 \pmod{4}$}, \medskip \\

	                 \frac{4(3k-2)}{(3k-1)^2} & \text{if $k \equiv 1 \pmod{4}$}, \medskip \\
	                 
	                 \frac{4(3k + 4)}{(3k+2)^2}  & \text{if $k \equiv 2 \pmod{4}$}, \medskip\\	              
	                 
	               \frac{4(3k + 2)}{(3k+1)^2} & \text{if $k \equiv 3 \pmod{4}$}.
\end{cases}
$$
Moreover, this maximum  is attained only at
$$
\begin{cases}
	   \pmb{\lambda} = \Big( \underbrace{\frac{2}{k}, \cdots, \frac{2}{k}}_{\frac{k}{4}}, \underbrace{\frac{2}{3k}, \cdots,  \frac{2}{3k}}_{\frac{3k}{4}} \Big) & \text{if $k \equiv 0 \pmod{4}$ and $k>8$},  \smallskip \\
	   
	   \pmb{\lambda} \in \Big\{ \Big( \frac{1}{6},  \frac{1}{6},  \frac{1}{6},  \frac{1}{6},  \frac{1}{6},  \frac{1}{6}, 0, 0\Big), \, \Big( \frac{1}{4}, \frac{1}{4}, \frac{1}{12},  \frac{1}{12}, \frac{1}{12}, \frac{1}{12}, \frac{1}{12}, \frac{1}{12}\Big) \Big\}& \text{if $k=8$}, \smallskip \\
	   
             \pmb{\lambda} = \Big( \underbrace{\frac{6}{3k-1}, \cdots,  \frac{6}{3k-1}}_{\frac{k-1}{4}}, \,  \underbrace{\frac{2}{3k-1}, \cdots,  \frac{2}{3k-1}}_{\frac{3k+1}{4}} \Big) & \text{if $k \equiv 1 \pmod{4}$}, \smallskip \\
             
             \pmb{\lambda} = \Big( \underbrace{\frac{6}{3k+2}, \cdots,  \frac{6}{3k+2}}_{\frac{k+2}{4}}, \,  \underbrace{\frac{2}{3k+2}, \cdots,  \frac{2}{3k+2}}_{\frac{3k-2}{4}} \Big) & \text{if $k \equiv 2\pmod{4}$}, \smallskip \\
                      
             \pmb{\lambda} = \Big( \underbrace{\frac{6}{3k+1}, \cdots,  \frac{6}{3k+1}}_{\frac{k+1}{4}}, \,  \underbrace{\frac{2}{3k+1}, \cdots,  \frac{2}{3k+1}}_{\frac{3k-1}{4}} \Big) & \text{if $k \equiv 3\pmod{4}$}.         
\end{cases}
$$

If $4 \le k \le 7$, then the maximum of  $F(\pmb{\lambda}) $ on $\CS_{k}$ is equal to $\frac{1}{k-2}$. Moreover, this maximum is attained only at $\pmb{\lambda} = \Big( \underbrace{\frac{1}{k-2}, \cdots,  \frac{1}{k-2}}_{k-2}, 0, 0 \Big)$.

\end{theorem}

\begin{proof}
Recall that the maximum of $F$  on $\CS_{k}$  is attained only at certain vertices of $\CS_{k}$. 

At type $1$ vertices $\pmb{\lambda} =  \Big( \underbrace{\frac{1}{i}, \cdots,  \frac{1}{i}}_{i}, \,  \underbrace{0, \cdots 0}_{k-i} \Big)$ where $k-2 \le i \le k$, we have $F(\pmb{\lambda}) = \frac{1}{i}$. So the maximum at type $1$ vertices is $M_1 = \frac{1}{k-2}$. 

At type $2$ vertices $\pmb{\lambda} = \Big( \underbrace{\frac{2}{k-1+i}, \cdots, \frac{2}{k-1+i}}_{i}, \underbrace{\frac{1}{k-1+i}, \cdots,  \frac{1}{k-1+i}}_{k-1-i},  0 \Big)$ where $1 \le i \le k-3$, we have $F(\pmb{\lambda}) = \frac{k -1 + 3i }{(k-1+i)^2}$. The derivative of the function $\frac{k -1 + 3x}{(k-1+x)^2}$ in $x$ is $\frac{k -1- 3x}{(k-1+x)^3}$, so its critical point is $x =\frac{k-1}{3}$. The maximum at type $2$ vertices is
$$
M_2 = \begin{cases}
	     \max (\frac{9(2k -4)}{(4k-6)^2}, \frac{9(2k -1)}{(4k-3)^2}) = \frac{9(2k -1)}{(4k-3)^2} & \text{if $k \equiv 0 \pmod{3}$}, \smallskip \\
	     
	      \frac{9}{8(k-1)}& \text{if $k \equiv 1 \pmod{3}$}, \smallskip \\
	      
             \max (\frac{9(2k -3)}{(4k-5)^2}, \frac{18k}{(4k-2)^2}) = \frac{9(2k -3)}{(4k-5)^2} & \text{if $k \equiv 2 \pmod{3}$}.       
\end{cases}
$$

At type $3$ vertices  $\pmb{\lambda} = \Big( \underbrace{\frac{3}{k+2i}, \cdots, \frac{3}{k+2i}}_{i}, \underbrace{\frac{1}{k+2i}, \cdots,  \frac{1}{k+2i}}_{k-i} \Big)$ where $1 \le i \le k-3$, we have $F(\pmb{\lambda}) = \frac{k + 8i }{(k+2i)^2}$. The derivative of the function $\frac{k + 8x}{(k+2x)^2}$ in $x$ is $\frac{4(k-4x)}{(k+2x)^3}$, so its critical point is $x =\frac{k}{4}$. The maximum at type $3$ vertices is
$$
M_3 = \begin{cases}
\frac{4}{3k} & \text{if $k \equiv 0 \pmod{4}$}, \smallskip \\

	                 \max (\frac{4(3k-2)}{(3k-1)^2}, \frac{4(3k + 6)}{(3k+3)^2}) = \frac{4(3k-2)}{(3k-1)^2} & \text{if $k \equiv 1 \pmod{4}$}, \smallskip \\
	                 
	                \max (\frac{4(3k -4)}{(3k-2)^2}, \frac{4(3k + 4)}{(3k+2)^2}) = \frac{4(3k + 4)}{(3k+2)^2}  & \text{if $k \equiv 2 \pmod{4}$}, \smallskip\\
	                
	                \max (\frac{4(3k-6)}{(3k-3)^2}, \frac{4(3k + 2)}{(3k+1)^2}) = \frac{4(3k + 2)}{(3k+1)^2} & \text{if $k \equiv 3 \pmod{4}$}. 
\end{cases}
$$

We now compare all the values of the function $F(\pmb{\lambda})$ at type $1$, $2$ and $3$ vertices. This comparison is divided into four cases as follows.

\smallskip
\underline{Case $1$}: $k \equiv 0 \pmod{4}$. 
\smallskip

If $k \ge 12$, then $M_3 = \frac{4}{3k} > \max (\frac{9(2k -1)}{(4k-3)^2}, \frac{9}{8(k-1)}, \frac{9(2k -3)}{(4k-5)^2} ) \ge M_2$ and $M_3 = \frac{4}{3k} > M_1 = \frac{1}{k-2}$. Hence the maximum is $\frac{4}{3k}$ which is attained at the type $3$ vertex $\Big( \underbrace{\frac{2}{k}, \cdots, \frac{2}{k}}_{\frac{k}{4}}, \underbrace{\frac{2}{3k}, \cdots,  \frac{2}{3k}}_{\frac{3k}{4}} \Big)$. 

If $k=8$, then $M_3 = \frac{1}{6} = M_1 > M_2 = \frac{13}{81} $. Hence the maximum is $\frac{1}{6}$ which is attained at the type $1$ vertex $\Big( \frac{1}{6},  \frac{1}{6},  \frac{1}{6},  \frac{1}{6},  \frac{1}{6},  \frac{1}{6}, 0, 0\Big)$ and the type $3$ vertex $\Big( \frac{1}{4}, \frac{1}{4}, \frac{1}{12},  \frac{1}{12}, \frac{1}{12}, \frac{1}{12}, \frac{1}{12}, \frac{1}{12}\Big)$.

If $k=4$, then $M_1 =  \frac{1}{2} > M_2 = \frac{3}{8} > M_3 = \frac{1}{3}$. Hence the maximum is $\frac{1}{2}$ which is attained at the type $1$ vertex $\Big( \frac{1}{2},  \frac{1}{2},  0, 0\Big)$.

\smallskip
\underline{Case $2$}: $k \equiv 1 \pmod{4}$.
\smallskip

If $k \ge 9$, then $M_3 = \frac{4(3k-2)}{(3k-1)^2}> \max (\frac{9(2k -1)}{(4k-3)^2}, \frac{9}{8(k-1)}, \frac{9(2k -3)}{(4k-5)^2} ) \ge M_2$ and $M_3 = \frac{4(3k-2)}{(3k-1)^2}> M_1 = \frac{1}{k-2}$. Hence the maximum is $\frac{4(3k-2)}{(3k-1)^2}$ which is attained at the type $3$ vertex $\Big( \underbrace{\frac{6}{3k-1}, \cdots,  \frac{6}{3k-1}}_{\frac{k-1}{4}}, \,  \underbrace{\frac{2}{3k-1}, \cdots,  \frac{2}{3k-1}}_{\frac{3k+1}{4}} \Big)$.

If $k=5$, then $M_1 =  \frac{1}{3} > M_2 = \frac{7}{25} > M_3 = \frac{13}{49}$. Hence the maximum is $\frac{1}{3}$ which is attained at the type $1$ vertex $\Big( \frac{1}{3},  \frac{1}{3},  \frac{1}{3}, 0, 0\Big)$.

\smallskip
\underline{Case $3$}:  $k \equiv 2 \pmod{4}$.
\smallskip

If $k \ge 10$, then $M_3 = \frac{4(3k + 4)}{(3k+2)^2} > \max (\frac{9(2k -1)}{(4k-3)^2}, \frac{9}{8(k-1)}, \frac{9(2k -3)}{(4k-5)^2} ) \ge M_2$ and $M_3 = \frac{4(3k + 4)}{(3k+2)^2} > M_1 = \frac{1}{k-2}$. Hence the maximum is $\frac{4(3k + 4)}{(3k+2)^2}$ which is attained at the type $3$ vertex $\Big( \underbrace{\frac{6}{3k+2}, \cdots,  \frac{6}{3k+2}}_{\frac{k+2}{4}}, \,  \underbrace{\frac{2}{3k+2}, \cdots,  \frac{2}{3k+2}}_{\frac{3k-2}{4}} \Big)$.

If $k=6$, then $M_1 =  \frac{1}{4} > M_2 = \frac{11}{49} > M_3 = \frac{11}{50}$. Hence the maximum is $\frac{1}{4}$ which is attained at the type $1$ vertex $\Big( \frac{1}{4},  \frac{1}{4},  \frac{1}{4},  \frac{1}{4}, 0, 0\Big)$.

\smallskip
\underline{Case $4$}: $k \equiv 3 \pmod{4}$.
\smallskip

If $k \ge 11$, then $M_3 = \frac{4(3k+2)}{(3k+1)^2}> \max (\frac{9(2k -1)}{(4k-3)^2}, \frac{9}{8(k-1)}, \frac{9(2k -3)}{(4k-5)^2} ) \ge M_2$ and $M_3 = \frac{4(3k+2)}{(3k+1)^2}> M_1 = \frac{1}{k-2}$. Hence the maximum is $\frac{4(3k+2)}{(3k+1)^2}$ which is attained at the type $3$ vertex $\Big( \underbrace{\frac{6}{3k+1}, \cdots,  \frac{6}{3k+1}}_{\frac{k+1}{4}}, \,  \underbrace{\frac{2}{3k+1}, \cdots,  \frac{2}{3k+1}}_{\frac{3k-1}{4}} \Big)$.

If $k=7$, then $M_1 =  \frac{1}{5} > M_3 = \frac{23}{121} > M_2 = \frac{3}{16}$. Hence the maximum is $\frac{1}{5}$ which is attained at the type $1$ vertex $\Big( \frac{1}{5}, \frac{1}{5},  \frac{1}{5},  \frac{1}{5},  \frac{1}{5}, 0, 0\Big)$.
\end{proof}

We now complete the proof of Theorem $1$. Since $\lambda_1 \le \lambda_{mn-2} + \lambda_{mn-1} +\lambda_{mn}$ for any APPT state $\rho$ on $\CH_A \otimes \CH_B$, the purity $\mathrm{Tr}(\rho^2)$ is bounded above by the maximum of the function $F(\pmb{\lambda})$ on $\CS_{mn}$. Theorem $1$ then follows from Theorem \ref{linear}  by taking $k=mn$.

\end{document}